\documentclass[a4paper]{article}

\usepackage{cite}
\usepackage[dvips]{graphicx}
\usepackage{latexsym}

\usepackage[all,cmtip]{xy}
\PassOptionsToPackage{hyphens}{url}\usepackage{hyperref}
\usepackage{epsfig,endnotes}

\usepackage{amsmath}
\usepackage{amsthm}
\usepackage{amssymb}
\usepackage{amsmath}
\usepackage{bm}
\usepackage{algpseudocode,amssymb,amsfonts,color,xspace}
\usepackage{pdfsync}
\usepackage{threeparttable}
\usepackage{color}
\usepackage[ruled, linesnumbered,lined,boxed,commentsnumbered]{algorithm2e}
\usepackage{url}

\usepackage[T1]{fontenc}
\usepackage[utf8]{inputenc}
\usepackage{authblk}

\usepackage{multirow}
\usepackage{threeparttable}

\usepackage{pgfplotstable}
\pgfplotstableset{
  empty cells with={---},
  every head row/.style={before row=\toprule,after row=\midrule},
  every last row/.style={after row=\bottomrule}
}
\pgfplotsset{compat=1.9}

\usepackage{array,booktabs}

\newcommand{\GF}[1]{\ensuremath{\mathbb{F}_{#1}}\xspace}

\newcommand{\TGF}[1]{\ensuremath{\mathbb{\widetilde{F}}_{#1}}\xspace}

\newcommand{\xor}{\xspace\texttt{XOR}\xspace}
\newcommand{\pshufb}{\texttt{PSHUFB}\xspace}
\newcommand{\vpshufb}{\texttt{VPSHUFB}\xspace}
\newcommand{\pclmul}{\xspace\texttt{PCLMULQDQ}\xspace}

\newcommand{\pointmul}{\texttt{pointmul}\xspace}

\newcommand{\lchbtfy}{\ensuremath{\texttt{FFT}_{\texttt{LCH}}}\xspace}
\newcommand{\ilchbtfy}{\ensuremath{\texttt{iFFT}_{\texttt{LCH}}}\xspace}

\newcommand{\bscvt}{\ensuremath{\texttt{BasisCvt}\xspace}}
\newcommand{\ibscvt}{\ensuremath{\texttt{iBasisCvt}\xspace}}

\newcommand{\novelpoly}{\emph{novelpoly}\xspace}

\newtheorem{theorem}{Theorem}[section]
\newtheorem{corollary}{Corollary}

\newtheorem{proposition}{Proposition}

\theoremstyle{definition}
\newtheorem{definition}[theorem]{Definition}
\newtheorem{myclaim}{Claim}

\newif\ifsubmission
\submissionfalse

\begin{document}

\title{Faster Multiplication for Long Binary Polynomials\\
}

\ifsubmission
\author{}
\institute{\vspace*{-1cm}}
\else

\author[1,3]{Ming-Shing~Chen}
\author[1]{Chen-Mou~Cheng}
\author[1,2]{Po-Chun~Kuo}
\author[2]{Wen-Ding Li}
\author[2,3]{Bo-Yin~Yang}
\affil[1]{Department of Electrical Engineering, National Taiwan University, Taiwan\\
\texttt{ \{mschen,doug,kbj\}@crypto.tw}
}
\affil[2]{Institute of Information Science, Academia Sinica, Taiwan\\
\texttt{ \{thekev,by\}@crypto.tw}
}
\affil[3]{Research Center of Information Technology and Innovation, Academia Sinica, Taiwan}

\fi

\maketitle

\begin{abstract}
  We set new speed records for multiplying long polynomials over
  finite fields of characteristic two.  Our multiplication algorithm
  is based on an additive FFT (Fast Fourier Transform) 
  by Lin, Chung, and Huang in 2014
  comparing to
  previously best results based on multiplicative FFTs.  Both
  methods have similar complexity for arithmetic operations 
  on underlying finite field;
  however, our implementation shows that the additive FFT has less overhead.
  For further optimization, we employ a tower field
  construction because the multipliers in the additive FFT
  naturally fall into small subfields, which leads to speed-ups using
  table-lookup instructions in modern CPUs.  Benchmarks show that our
  method saves about $40 \%$ computing time when multiplying
  polynomials of $2^{28}$ and $2^{29}$ bits comparing to previous
  multiplicative FFT implementations.

\end{abstract}




\section{Introduction}
\label{sec:introduction}

Multiplication for
long binary polynomials in the ring $\GF{2}[x]$, where $\GF{2}$
is the finite field with two elements,
is a fundamental 
problem in computer science.
It is needed to factor polynomials in number theory
\cite{vonzurGathen:1996:AFP:236869.236882}\cite{von2002polynomial}
and
is central to the modern Block Wiedemann algorithm
\cite{thome2002subquadratic}.  Block Wiedemann is in turn integral to
the Number Field Sieve \cite{aoki2007kilobit}, the critical attack
against the RSA cryptosystem.  Another application of Block Wiedemann
is the XL algorithm \cite{pxl}, which is a critical attack against
multivariate public-key cryptosystems \cite{MPKC}.

To the best of our knowledge, all currently
fast binary polynomial multiplication algorithms are based on
a Fast Fourier Transform (FFT) algorithm.
The FFT is an algorithm for efficiently
evaluating polynomials at subgroups in the underlying finite field.
%


\subsection{Previous works on multiplying binary polynomials}
There have been
good reviews of algorithms and implementations for multiplying
binary polynomials in \cite{gf2x} and \cite{DBLP:conf/issac/HarveyHL16}. 

In general, these previous methods of multiplication were based on
``multiplicative'' FFTs, evaluating 
polynomials at multiplicative subgroups formed by roots of unity.
Since the sizes of multiplicative groups existing in \GF{2^k} are
restricted, multiplicative FFTs in the binary field
are somewhat more difficult than that over the reals and the complex.
 These multiplicative FFTs, evaluating polynomials at $n$ points, 
reach their best efficiency
only when $n$ is a size of a multiplicative subgroup.

In \cite{gf2x}, Brent \emph{et al.} implemented mainly the
Sch\"onhage \cite{schonhage1977} algorithm for long polynomials
with complexity $O(n \log n \log\log n)$ of field operations.  
Sch\"onhage's algorithm is
based on ternary FFTs over binary finite fields.  In their
implementation, the optimal number of evaluation points is $3^k$.

Harvey, van der Hoeven, and Lecerf \cite{DBLP:conf/issac/HarveyHL16} presented multiplication
using an FFT of a mixed radix approach.
They applied several discrete Fourier transforms(DFTs) for different
input sizes, e.g., Cooley-Tukey\cite{cooley1965algorithm} for the largest
scale DFT.
In particular, they need
to find a suitable finite field which is of a size close to a machine
word and simultaneously allows both abundant multiplicative subgroups.
They proposed $\GF{2^{60}}$ which elegantly satisfies these conditions.

\subsection{Recent Progress: Additive FFTs}
Following Cantor~\cite{Cantor:1989}, alternative methods are developed to
evaluate polynomials at points that form an \emph{additive} subgroup 
in a field of characteristic 2.  These methods are called
``additive FFTs''.

Cantor presented a basis, termed ``Cantor basis'' in the literature,
for constructing a finite field as well as an FFT over the field.
His FFT works with the complexity of $O(n \log n)$ multiplications and
$O(n \log^{\log_2 3} n)$ additions(\xor) for evaluating $n=2^m$ elements.

In 2010, Gao and Mateer~\cite{gm-afft} presented an additive FFT
(heretofore ``GM FFT'') over \GF{2^k}, where the evaluation points are
an additive subgroup of size $2^k$ in the underlying GF.
  The additive subgroups are easier to form than multiplicative subgroups
in the fields of characteristic 2.  
However, the complexity in GM FFT is
$O(n \log^2 n)$ \xor operations for evaluating a polynomial at
$n = 2^m$ points in general.
It can be optimized to $O(n \log n \log \log n)$ only when $m$ is a
power of 2, and the polynomials in this case are represented in a
special polynomial basis introduced by Cantor~\cite{Cantor:1989}.
In 2014, Bernstein and Chou~\cite{auth256} presented an efficient
implementation of the GM FFT. 

In 2014, Lin, Chung, and Han~\cite{lch-afft}
proposed a more general variant (``LCH
FFT''), which uses a different polynomial basis than
the standard one.  
The polynomial basis results in a more regular structure in the
butterfly stage.  When representing the underlying finite field
in a Cantor basis, the butterfly stage is the same as the optimized
GM FFT and Bernstein-Chou.  
In the subsequent work \cite{Lin:BasisCvt:2016}, they
presented a method for converting the polynomial from standard basis.
The complexity for the conversion is $O(n \log n \log \log n)$ \xor operations
for $n$ being any power of two, which reaches the same complexity as multiplicative
FFTs.

More details of the LCH FFT are reviewed in Sec.~\ref{sec:lchfft}.

%
%

\subsection{Our Contributions}
In this paper, we present a faster method of multiplication for
long binary polynomials based on the recently developed LCH FFT
comparing to previously multiplicative-FFT-based algorithms.
From the faster results of our additive-FFT-based implementations,
we confirm that the recent development of additive
FFTs helps the multiplication for binary polynomials.
%

Our implementation is faster than previous multiplicative FFT codes for two reasons.
First, the FFT in our algorithm uses simple binary
butterflies stages leading to $\frac12 n \log n$ multiplications
for $n$ evaluation points,
compared to (for example) the ternary FFT used in Sch\"onhage's algorithm
which leads to $\frac43 n \log_3 n$ multiplications.
This factor confers a 10\%--20\% advantage over multiplicative-FFT-based
implementations and will be discussed in Sec.~\ref{sec:simple:mul}.
Second, we exploit the fact that the multipliers in the additive FFT are in
subfields, which reduces the average time taken per multiplication.
We will discuss the method in Sec.~\ref{sec:gf-ver2} and \ref{sec:tower:mul}.
The overall improvement is 10\%--40\% over previous implementations.

%

 


\section{Preliminaries}
\label{sec:preliminary}

\subsection{Multiplying with Segmentation of Binary Polynomials}
\label{sec:mul-seg-fft}
In this section,
we discuss the general method of multiplications for
long binary polynomials.

Suppose we are multiplying two polynomials $a(x) = a_0 + a_1 x + \cdots + a_{d-1} x^{d-1}$
and $b(x) = b_0 + \cdots + b_{d-1} x^{d-1} \in \GF{2}[x]$.
The polynomials are
represented in bit sequence with length $d$.
The standard \emph{Kronecker segmentation} for multiplying binary polynomials
is performed as follows:
\begin{enumerate}
\item Partition the polynomials to $w$-bits blocks. There are $n = \lceil d/w \rceil$ blocks.
\begin{multline*}
 a(x) = a_0 + a_1 x + \cdots  + a_{d-1} x^{d-1} \\
 \rightarrow 
 (a_0 + \cdots + a_{w-1}x^{w-1}) + ( a_{w} + \cdots + a_{2w-1} x^{w-1}) x^w + \cdots + (\cdots) x^{w(n-1)} \enspace .
\end{multline*}
\item Define $\GF{2^{2w}} := \GF{2}[z]/(g(z))$, with $g(z)$
  an irreducible polynomial of degree $2w$.
Let $\psi$ map $a(x)$ to $a'(y) \in \GF{2^{2w}}[y]$ 
(and similarly $\psi(b(x))=b'(y)$) :
\begin{eqnarray*}
  a'(y) &:=& a'_0 + a'_1 y + \cdots + a'_{n-1} y^{n-1} \in \GF{2^{2w}}[y], \\
  b'(y) &:=& b'_0 + b'_1 y + \cdots + b'_{n-1} y^{n-1} \in \GF{2^{2w}}[y], 
\end{eqnarray*}
such that  $a'_0 = (a_0 + a_1 z + \ldots + a_{w-1}z^{w-1}), a'_1 =
( a_{w} +\ldots + a_{2w-1} z^{w-1})$, \ldots, $a'_{n-1}=( a_{(n-1)w} + a_{(n-1)w+1} z + \ldots + a_{nw-1} z^{w-1})$
 and same for  $b_j$.
\item Calculate $c'(y) = a'(y) \cdot b'(y) = c'_0+c'_1 y + \cdots \in \GF{2^{2w}}[y]$ (using FFTs).
\item Map $z\mapsto x$ and $y\mapsto x^{w}$. Then, collect terms and
  coefficients to find the result of multiplication of binary
  polynomials.  We need to add together at most 2 coefficients at any power.
\end{enumerate}

\paragraph{FFT-based Polynomial multiplication.}
It is well known that 
polynomial multiplication can be done using FFT\cite{Cormen:2009:IntroAlgo}.
To multiply two degree-$(n-1)$ polynomials $ a'(y)$ and $b'(y) \in \GF{2^{2w}}[y]$
with FFT algorithms, the standard steps are as follows:
\begin{enumerate} 
\item (\texttt{fft}) Evaluate $a'(y)$ and $b'(y)$ at $2n$ points by an FFT algorithm.
\item (\texttt{pointmul}) Multiply the evaluated values pairwise together.
\item (\texttt{ifft}) Interpolate back into a polynomial of degree $\le 2n-1$ by the inverse FFT algorithm.
\end{enumerate}
The complexity of polynomial multiplication is the same as the FFT in use.

\subsection{Alternative Representations of Finite Fields}
\label{sec:gf-arithmetic}
The field of two elements, denoted as \GF{2}, is the set $\{ 0,1\}$.
The multiplication of \GF{2} is logic \texttt{AND} and addition is
logic \xor.  In this paper, every field will be an algebraic extension of \GF{2}.

We will apply interchangeable representations for each used finite
field (or Galois field, GF).  The illustrative example
is the field of $2^{128}$ elements, denoted by
$F_7$ in \cite{Cantor:1989} and \cite{gf2x}. 
We will switch representations during computation to
achieve a better efficiency for field multiplication.
All fields of the same size are isomorphic and the cost of changing
representation is a linear transformation. 

\subsubsection{The Irreducible Polynomial Construction of \GF{2^{128}}}
\label{sec:gf-ver1}
We choose the basic working field to be the same as in AES-GCM, denoted
as $\GF{2^{128}}$:
\[
\GF{2^{128}} := \GF{2}[x]/\left( x^{128} + x^7 + x^2 + x + 1\right) \enspace .
\]
An element in $\GF{2^{128}}$ is represented as a binary polynomial
of degree $< 128$.
The $\GF{2^{128}}$ can also be a linear space of dimension 128 with the basis $(x^i)_{i=0}^{127}$.

\paragraph{The cost of multiplication for $\GF{2^{128}}$.}
There are hardware instructions for multiplying small binary polynomials
which fits for the multiplication of $\GF{2^{128}}$ in many platforms.
\pclmul is a widely used instruction for multiplying 64-bit binary polynomials in x86.
Since one \pclmul performs $64 \times 64 \rightarrow 128$ bits,
the multiplication of \GF{2^{128}} costs roughly 5 \pclmul(3 for multiplying $128$-bit
polynomials with Karatsuba's method and 2 for reducing the $256$-bit result
back to $128$ bits with linear folding).
More details about multiplications of \GF{2^{128}} can be found in \cite{pclmul:gcm}.

%
%

\subsubsection{Cantor Basis for Finite Field as Linear Space}
\label{sec:gf-ver3}

Gao and Mateer presented an explicit construction of Cantor Basis for finite field
in \cite{gm-afft}.
The Cantor basis $(\beta_i)$ satisfies $\beta_0=1,
\beta_i^2 + \beta_i = \beta_{i-1}$ for $i>0$. 


\begin{definition} With respect to the basis $(\beta_i)$, let $\phi_{\beta}(k):=\sum_{j=0}^{m-1} 
b_j \beta_j$ be the \emph{field element represented by $k$ under $(\beta_i)$}
when the binary expansion of $k=\sum_{j=0}^{m-1} b_j 2^j$ with $b_j \in \{ 0,1\}$.
\end{definition}

\begin{definition}\label{def:si}
Given a basis $(\beta_i)_{i=0}^{m-1}$ in the base field,  its 
\underline{sequence of subspaces} is $V_i:=\mathrm{span}\{\beta_0,\beta_1,\ldots, \beta_{i-1}\}$.
Its \underline{subspace vanishing polynomials} $(s_i)$ are,
\[
 s_i(x) := \prod_{a \in V_i} (x-a) \enspace.
\]
\end{definition}
Note that $V_i$ is a field with linear basis $(\beta_j)_{j=0}^{i-1}$
only when $i$ is power of $2$.
Since $\dim V_i=i$, one can see that $\deg(s_i(x)) = 2^i$.

From \cite{Cantor:1989} and \cite{gm-afft}, vanishing polynomials $s_i(x)$
w.r.t.\ the Cantor basis $(\beta_i)$ has the following useful properties:
\begin{itemize}
\item (linearity) $s_i(x)$ contains only monomials in the form $x^{2^m}$.
\item (minimal two terms) $s_i(x) = x^{2^i} + x$ iff $i$ is a power of $2$. 
In the cases of $i=2^k$ ,\ $V_{2^k}$ are fields and
$\{\beta_{2^k},\beta_{2^k+1},\ldots,\beta_{2^{k+1}-1}\} \subset \GF{2^{2^{k+1}}}
\backslash \GF{2^{2^k}}$.
\item (recursivity) $s_i(x) = s_{i-1}^2(x) + s_{i-1}(x) = s_1(s_{i-1}(x))$; $s_{i+j}(x)=s_i(s_j(x))$.
\end{itemize}
If
$k=2^{i_0}+2^{i_1}+\cdots+2^{i_j}$, where $i_0<i_1<\cdots<i_j$, then we
can write
$s_k(x) = s_{2^{i_0}}(s_{2^{i_1}}(\cdots(s_{2^{i_j}}(x))\cdots))$.
Therefore, every $s_i$ is a composition of functions
which only has two terms (see Table~\ref{tab:si}).

\paragraph{Evaluating $s_i(x)$ in Cantor basis}
Computing $s_i(x)$ in the Cantor basis is very fast since
the representation
  of $s_i(\alpha)$ would exactly be that of $\alpha$ shifted right by
  $i$ bits, or $s_i(\phi_\beta(j))=\phi_\beta(j \gg i)$. 
For example, we have $s_i(\beta_i) = \beta_0 = 1$.


\subsection{The Lin-Chung-Han (LCH) FFT}
\label{sec:lchfft}
In this section,
we introduce how to evaluate a degree 
$(n-1)$ polynomial at $n$ points 
in a binary field with the LCH FFT. 
We assume that $n$ is a power of $2$.
The polynomial is padded with $0$ for high-degree 
coefficients if the actual degree of polynomial is not $n-1$.
%

The LCH FFT requires that the evaluated polynomial is converted into a
particular kind of basis, called \novelpoly basis.


\subsubsection{The \novelpoly basis}
The polynomial basis used in LCH FFT was presented in \cite{lch-afft}.
The \novelpoly basis for polynomials must be distinguished from bases for field.
Although the LCH FFT is independent of underlying basis of
finite field, we assume that the field is in Cantor basis for simplicity.

\begin{definition}\label{def:novelpoly}
Given the Cantor basis $(\beta_i)$ for the base field and its 
vanishing polynomials $(s_i)$,
define the \underline{\novelpoly basis} w.r.t.\ $(\beta_i)$ to be the polynomials $(X_k)$
\[
X_k(x):= \prod \left(s_i(x)\right)^{b_i} \quad 
\mbox{ where } k=\sum b_i 2^i \mbox{ with } b_i \in \{ 0,1\}\enspace .
\]
I.e.,
$X_k(x)$ is the product of all $s_i(x)$ where the $i$-th bit
of $k$ is set.
\end{definition}
Since $\deg(s_i(x)) = 2^i$, clearly  $\deg(X_k(x)) =k$.  

To perform LCH FFT,
the evaluated polynomial $f(x)$ has to be converted into the form 
$f(x) = g(X) = g_0 + g_1 X_1(x) + \ldots + g_{n-1} X_{n-1}(x)$.  
%

\subsubsection{LCH's Butterfly}
The evaluation of a polynomial in the \novelpoly basis 
can be done through a ``Butterfly'' process, denoted as
\lchbtfy.\footnote{ following \cite{lch-afft}, which calls the
  butterfly an FFT.}

The general idea of evaluating $f(x)$ at all points of $V_k$ is to
divide $V_k$ into the two sets $V_{k-1}$ and 
$V_k\backslash V_{k-1} = V_{k-1}+\beta_{k-1}:=  \{x+\beta_{k-1}: x\in V_{k-1} \}$. 
Since $s_k(x)$ is linear, 
evaluations at $V_{k-1}+\beta_{k-1}$
can be quickly calculated with the evaluations at $V_{k-1}$
and the butterfly process.
It is a divide-and-conquer process that the polynomial $f(x) = g(X)$ can be
expressed as two half-sized polynomials $h_0(X)$ and $h_1(X)$ with
$g(X) = h_0(X) + X_{2^{\lceil \log n \rceil-1}} (x) h_1(X)$.
%

\lchbtfy is detailed in Algorithm~\ref{alg:fft}.
The \lchbtfy evaluates the converted polynomial $g(X)$ at points $V_{\log^{n}} + \alpha$.
Line 5 and 6 perform the butterfly process (see
Fig.~\ref{fig:butterfly-detail} and Fig.~\ref{fig:example:butterfly:polymul}).
Inverse $\lchbtfy$ simply performs the butterflies in reverse.

\begin{algorithm}[!tbh]
\SetKwFunction{nfft}{\ensuremath{\texttt{FFT}_{\texttt{LCH}}}}
\SetKwInOut{Input}{input}
\SetKwInOut{Output}{output}
\nfft{$f(x) = g(X), \alpha $} : \\
\Input{ a polynomial: $ g(X) = g_0 + g_1 X_1(x) + ... + g_{n-1} X_{n-1}(x) $ .\\
        a scalar: $ \alpha \in \GF{} $ .\\}
\Output{ a list: $ [ f( 0 + \alpha ),f( \phi_u(1) + \alpha ), \ldots , f( \phi_u(n-1) + \alpha)  ] $ .}
\BlankLine
\lIf { $ \deg( f(x) ) = 0 $ }{ return $[ g_0 ]$ }
Let $k \gets \text{Max}(i)$ s.t. $2^i \leq n-1$ . \\
Let $ g(X) = p_0(X) + X_{2^k} \cdot p_1(X) = p_0(X) + s_k(x) \cdot p_1(X) $. \\
$h_0(X) \gets p_0(X) + s_k(\alpha) \cdot p_1(X) $. \\
$h_1(X) \gets h_0(X) + s_k(\beta_k) \cdot p_1(X) $. \textbf{//} $s_k(\beta_k) = 1$ in Cantor basis. \\
return $[$ \nfft{$  h_0(X) , \alpha $},\nfft{$ h_1(X) , \beta_k + \alpha $} $]$
\caption{LCH FFT in \novelpoly.}
\label{alg:fft}
\end{algorithm}

We note there are two multipliers $s_k(\alpha)$ and $s_k(\beta_k)$ in the \lchbtfy and
$s_k(\beta_k) = 1$ in Cantor basis, avoiding one multiplication.
This constant reduction of multiplications won't affect the asymptotic complexity; however,
one can not bear the extra multiplication in practice.
Although \lchbtfy is applicable to any basis of field, the
choice for a practitioner might be Cantor basis only.

%

\subsection{Conversion to \novelpoly Basis w.r.t.\ Cantor Basis}
\label{sec:basis-conversion}

%

The evaluated polynomial has to be in \novelpoly basis for performing \lchbtfy.
We review the conversion algorithms in this section.
The fast conversion relies on the simple form of $(s_i)$ w.r.t. Cantor basis.

\cite{auth256} converts $f(x)$ to $g(X)$ by finding the
largest $i$ such that $2^i < \deg f$, and then divide $f(x)$ by
$s_i(x)$ to form $f(x) = f_0(x) + s_i(x) f_1(x)$.  Recursively divide
$f_0(x)$ and $f_1(x)$ by lower $s_i(x)$ and eventually express $f(x)$
as a sum of non-repetitive products of the $s_i(x)$, which is the
desired form for $g(X)$.  We know the division comprise only \xor
operations since the coefficients of $s_i(x)$ are always $1$ in Cantor basis.
Therefore the complexity of division by $s_i(x)$ depends on the number
of terms in $s_i(x)$.  This is
functionally equivalent to the Cantor transform.

\begin{table}
\caption{Variable Substitution of $s_i(x)$}
\label{tab:si}
\footnotesize
\begin{tabular}{c| r| l | l }
$ s_0(x) $ & $ x $         &   &  \\
$ s_1(x) $ & $ x^2 + x $   &  &  \\
\hline
$ s_2(x) $ & $ x^4 + x $   &  $  = s_{2}(x) = y$ & \\
$ s_3(x) $ & $ x^8 + x^4 + x^2 + x $ &  $ = s_{1}(y) = y^2 + y$ & \\
\hline
$ s_4(x) $ & $ x^{16} + x $     &   $ = s_{4}(x) = z$ &  \\
$ s_5(x) $ & $ x^{32} + x^{16} + x^2 + x $  &  $ = s_{1}(z) = z^2 + z $ & \\
$ s_6(x) $ & $ x^{64} + x^{16} + x^4 + x $  &  $ = s_{2}(z) = z^4 + z $ & $= s_{6}(x) = w $\\ 
$ s_7(x) $ & $ x^{128} + x^{64} + \cdots + x^2 + x $  &  $=s_{3}(z) = z^8 + z^4 + z^2 + z $ & $=s_{1}(w) = w^2 + w $\\
\end{tabular}
\end{table}



\cite{Lin:BasisCvt:2016} does better by
\begin{enumerate} \setlength{\itemsep}{0in}
\item finding the largest $2^{2^i}$ such that $2^{2^j}<\deg f$ and then do
  variable substitution (Alg.~\ref{alg:varSubs}) to express $f$ as a
  power series of $s_{2^i}$.
\item Recursively express the series in $s_{2^i}$ as a series in $X_j(s_{2^i})$,
  where $j<2^{2^i}$.  
\item Recursively express each coefficient of $X_j(s_{2^i})$ (which is
  a polynomial in $x$ of degree $< 2^{2^i}$) as a series in $X_k$, where $k<2^{2^i}$.
\end{enumerate}

%
%
\begin{algorithm}[!h]
\SetKwFunction{vs}{VarSubs}
\SetKwInOut{Input}{input}
\SetKwInOut{Output}{output}
\vs{ $f(x), y $ } : \\
\Input{ Two polynomials: $ f(x) = f_0 + f_1 x + ... + f_{n-1} x^{n-1} $ and $ y = x^{2^i} + x$ .}
\Output{ $ h(y) = h_0(x) + h_1(x)y + \cdots + h_{m-1}(x) y^{m-1} $ .}
\BlankLine
\lIf { $ \deg( f(x) ) < 2^i $ }
{
  return $h(y) \gets f(x)$
}
Let $ k \gets \text{Max}(2^j)$ where $j \in \mathbb{Z} \quad \text{s.t. } \deg( (x^{2^i} + x)^{2^j} ) \leq \deg( f(x) )$  .\\
Let $y^k \gets x^{k2^i} + x^k$. \\
Compute $ f_0(x) + y^k \cdot f_1(x) = f(x)$ by dividing $f(x)$ by $x^{k 2^i } + x^k$. \\
\textbf{//} Note that this is done by repeatedly subtracting. \\
return \vs{ $f_0(x), y $ } $\, + y^k \cdot $ \vs{ $f_1(x), y $ } .
\caption{Variable Substitution}
\label{alg:varSubs}
\end{algorithm}

The detail of basis conversion is given in Algorithm~\ref{alg:changeBasis2} and
an example is given in Appendix~\ref{appendix:basis-conversion}.
Note that the algorithms 
rely on the simple form of $(s_i)$ instead of field representations of coefficients.
\begin{algorithm}[!h]
\SetKwFunction{vs}{VarSubs}
\SetKwFunction{cvt}{BasisCvt}
\SetKwInOut{Input}{input}
\SetKwInOut{Output}{output}
\cvt{$f(x)$} : \\
\Input{ $ f(x) = f_0 + f_1 x + ... + f_{n-1} x^{n-1} $ .}
\Output{ $ g(X) = g_0 + g_1 X_{1}(x) + ... + g_{n-1} X_{n-1}(x) $ .}
\BlankLine
\lIf { $ \deg( f(x) ) \leq 1 $ }
{
  return $g(X) \gets f_0 + X_1 f_1$
}
Let $ k \gets \text{Max}(2^i)$ where $i \in \mathbb{Z} \quad \text{s.t. } \deg (s_{2^i}(x)) \leq \deg (f(x)) $ .\\
Let $y \gets s_k(x)$. \\
$h(y) = h_0(x) + h_1(x)y + \cdots + h_{m-1}(x) y^{m-1} \gets $ \vs{ $f(x),y$ } .\\
$h'(Y) = q_0(x) + q_1(x) X_{2^k} + \cdots + q_{m-1}(x) X_{(m-1)\cdot 2^{k}} \gets $ \cvt{ $h(y)$ } . \\
\ForEach{ $q_i(x) \text{ in } h'(Y)$ } { Compute $g_i(X) \gets $ \cvt{ $q_i(x)$ } . } 
return $ g(X) = g_0(X) + g_1(X) X_{2^k} + ... + g_{n-1}(X) X_{(m-1)\cdot 2^k} $ \\
\caption{Basis conversion: monomial to \novelpoly w.r.t\ Cantor.}
\label{alg:changeBasis2}
\end{algorithm}

%



\section{Binary Polynomial Products with Additive FFT}
\label{sec:method}

We can have a fast multiplication for binary polynomials simply
by applying LCH FFT with the Cantor basis as the underlying FFT
in the general method of Sec.~\ref{sec:mul-seg-fft}.

Besides the straightforward method, we also present a faster
algorithm by accelerating the field multiplication
in the FFT.
The acceleration relies on a special tower field representation,
 making all multipliers in butterflies short.

\subsection{A Simple Method of Multiplying Binary Polynomials}
\label{sec:simple:mul}

A simple version of our multiplication for binary polynomials is
to keep the working field in the representation of polynomial basis
and apply LCH FFT with evaluation points in Cantor basis to
the general multiplication in Sec.~\ref{sec:mul-seg-fft}.
The details of the straightforward algorithm is presented 
in Alg.~\ref{alg:bitpolymul-simple}.

For more details of \lchbtfy, we
choose $w=64$ and use $\GF{2^{128}}$ as our base field.
The evaluated points are $\{ \phi_\beta(0),\ldots,\phi_\beta(2n-1)\}$, which
can be seen on line~6 in Alg.~\ref{alg:bitpolymul-simple}.
The multiplier $s_k(\alpha)$ is calculated in Cantor basis and then switched
to its representation in $\GF{2^{128}}$ with a linear map (field isomorphism).
Although the multipliers in the Cantor basis are short numbers,
they are random-looking 128-bit polynomials in $\GF{2^{128}}$.
The field multiplication in $\GF{2^{128}}$ is performed with 5 \pclmul (cf.~
Sec.~\ref{sec:gf-ver1}).
The other multiplier $s_k(\beta_k) = 1$ in the Cantor basis.

%
%
%

\begin{algorithm}[!tbh]
\SetKwFunction{bpolymul}{\ensuremath{\texttt{binPolyMul}}}
\SetKwFunction{split}{\ensuremath{\texttt{Split}}}
\SetKwFunction{combine}{\ensuremath{\texttt{InterleavedCombine}}}
\SetKwFunction{changerep}{\ensuremath{\texttt{changeRepr}}}
\SetKwFunction{bc}{\ensuremath{\texttt{BasisCvt}}}
\SetKwFunction{nfft}{\ensuremath{\texttt{FFT}_{\texttt{LCH}}}}
\SetKwFunction{ibc}{\ensuremath{\texttt{iBasisCvt}}}
\SetKwFunction{infft}{\ensuremath{\texttt{iFFT}_{\texttt{LCH}}}}
\SetKwInOut{Input}{input}
\SetKwInOut{Output}{output}
\bpolymul{ $a(x),b(x)$ } : \\
\Input{ $a(x) , b(x) \in \GF{2}[x] $ .}
\Output{ $ c(x) = a(x)\cdot b(x) \in \GF{2}[x]$ .}
\BlankLine
$f_a(x) \in \GF{2^{w}}[x] \gets \split( a(x) )$ . \\
$f_b(x) \in \GF{2^{w}}[x] \gets \split( b(x) )$ . \\
$g_a(X) \in \GF{2^{w}}[X] \gets \bc( f_a(x) ) $ . \\
$g_b(X) \in \GF{2^{w}}[X] \gets \bc( f_b(x) ) $ . \\
$[ f_a(0),\ldots,f_a( \phi_\beta(2n-1) )] \in \GF{2^{2w}}^{2n} \gets \nfft( g_a(X) , 0 ) $. \\
$[ f_b(0),\ldots,f_b( \phi_\beta(2n-1) )] \in \GF{2^{2w}}^{2n} \gets \nfft( g_b(X) , 0 ) $. \\
$ [ f_c(0),\ldots,f_c(\phi_\beta(2n-1)) ] \in \GF{2^{2w}}^{2n} \gets 
  [f_a(0) \cdot f_b(0),\ldots,f_a(\phi_\beta(2n-1))\cdot f_b(\phi_\beta({2n-1})) ] $ \\
$ g_c(X) \in \GF{2^{2w}}[X] \gets \infft( [ f_c(0), \ldots , f_c( \phi_\beta({2n-1}) )] ) $ . \\
$ f_c(x) \in \GF{2^{2w}}[x] \gets \ibc( g_c(X) ) $ . \\
$ c(x) \in \GF{2}[x] \gets \combine( f_c(x) ) $ . \\
return $ c(x) $.
\caption{Simple multiplications for binary polynomials.}
\label{alg:bitpolymul-simple}
\end{algorithm}

\paragraph{Advantages of the Additive FFT}
We can expect that the simple structure of the LCH FFT leads to a
lower complexity.  One way is to count the butterflies.  LCH FFT has a
binary structure, which means at each of $\log n$ layers there are
$n/2$ butterflies for a total of $\frac12 n\log n$ multiplications.
Considering a ternary FFT instead, there will be $\log_3 n$ layers
of analogous structure to butterflies, in number $n/3$ each.  At
each of these structures, one has to make four multiplications for a
total of $\frac43 n \log_3 n$ multiplications.  All else
being equal, the multiplicative complexity of binary
structure of the additive FFT is about $1.68$ times lower than that of
the ternary FFT, which is used in the Sch\"ohage-like algorithm in
\cite{gf2x}.  Similarly we hold an advantage over the even more
complex FFT method in \cite{DBLP:conf/issac/HarveyHL16}.

\paragraph{Results} Please refer to Tab.~\ref{tab:bitpolymul} in Sec.~\ref{sec:results}.
Simply using an additive FFT confers a 10\%--20\% advantage over 
state-of-the-art libraries in \cite{DBLP:conf/issac/HarveyHL16,gf2x}.

\subsection{The Tower Construction for Binary Finite Fields}
\label{sec:gf-ver2}
%
We consider this sequence of extension fields. 
{\footnotesize
\[
\begin{array}{ rlrl}
 \GF{4} = \TGF{2^{2}} := & \GF{2}[x_1]/(x_1^2 + x_1 + 1), &
  \TGF{2^{32}} := & \TGF{2^{16}}[x_5]/(x_5^2+x_5+ \prod_{i=1}^{4} x_i), \\
 \GF{16} = \TGF{2^{4}} := & \GF{4}[x_2]/(x_2^2+x_2+x_1), &
  \TGF{2^{64}} := & \TGF{2^{32}}[x_6]/(x_6^2+x_6+ \prod_{i=1}^{5} x_i ), \\
 \GF{256} = \TGF{2^{8}} := & \GF{16}[x_3]/(x_3^2+x_3+x_2x_1), &
 \TGF{2^{128}} := & \TGF{2^{64}}[x_7]/(x_7^2+x_7+ \prod_{i=1}^{6} x_i), \\
 \TGF{2^{16}} := & \GF{256}[x_4]/(x_4^2+x_4+x_3x_2x_1), &
  \TGF{2^{256}} := & \TGF{2^{128}}[x_8]/(x_8^2+x_8+ \prod_{i=1}^{7} x_i). \\
\end{array}
\]
} Thus, decimal subscripts or a tilde 
denotes the field is in tower
representation.  We can now define a basis for \TGF{2^{256}} as a
vector space over \GF2.

\begin{definition}
\label{def:v:basis}
$v_k := \prod_{j=0}^{m-1} x_{j+1}^{b_j}$ where $k:=\sum_{j=0}^{m-1} b_j 2^j$ with $b_j \in \{ 0,1\}$.
\end{definition}
By definition, the sequence  
$(v_0,v_1,v_2,\allowbreak v_3,v_4,v_5,\ldots) := (1,x_1,x_2, \allowbreak x_2x_1,x_3
,\allowbreak x_3x_1, \ldots)$. 
%
%
Henceforth \emph{$(v_k)$ will be our default basis unless otherwise specified.}

%

\begin{definition}
$\overline{\imath}:=\phi_v(i)$ is the element of \TGF{2^k} represented by $i$.
Numbers in hex such as \texttt{0x1f} also denote the
representatives under the basis $(v_i)$.
%
\end{definition}
Hence the sequence $(v_k)$ can also be written as 
$(v_0,v_1,v_2,\allowbreak v_3,\ldots) := (1,\overline{2},\overline{4},\allowbreak \overline{8},\ldots)$ or
$(1,\texttt{0x2},\texttt{0x4},\allowbreak \texttt{0x8},\ldots)$.
For example, $ x_2 x_1 + x_2 + x_1 + 1 = v_0 + v_1 + v_2 + v_3 \in \GF{16}$ is denoted
as $\overline{15}$ or \texttt{0xf}.  
Under this notation, we can order two elements in the tower field and
thus define ``big'' or ``small'' by comparing their representative
numbers.
%
We can see that the representation of
each field is embedded in the lower bit(s) of the field that is twice
as wide.
For example, the elements $\{ \texttt{0x00}, \texttt{0x01},\ldots,\texttt{0x0f} \}$ in \GF{256}
form the subfield \GF{16}.
Therefore, a ``small'' number in the tower field is
often in a subfield.

\subsubsection{Compatibility between Tower and Cantor Bases}

We discuss the ``compatibility'' between tower and Cantor bases in this section.
It is clear the basis $(v_0,v_1,\ldots)$ is not a Cantor basis since
$v_3^2 + v_3 \neq v_2 $.
However, 
\begin{myclaim}
\label{lemma:space:equal}
  The sequence of subspaces $(V_k)$ and subspace vanishing polynomials
  $(s_i)$ are the same w.r.t.\ the Cantor basis $(\beta_j)$ and the
  tower basis $(v_j)$.
\end{myclaim}

We first assume that
$V_k := \mathrm{span}(v_0,\ldots,\, v_{k-1})$ and
$s_k(x)=\prod_{c\in V_k}(x-c)$ and we will show that these are the same
as those from a Cantor basis.
Note again that $V_k$ is a field only for $k=2^m$ a power of two, and only
in this case we have $V_k =\TGF{2^{k}}= \TGF{2^{2^{m-1}}}[x_m]/(x_m^2+x_m+v_{2^m-1})$.

\begin{proposition}
\label{lemma:subPolyEval}
$s_k( v_k ) = 1$.
\end{proposition}
\textbf{Proof:}  See Appendix~\ref{sec:prooft}.\\

\begin{corollary}[recursivity]
With respect to $(v_i)$,\ \(s_{k+1}(x) = s_k(x) \prod_{c\in V_k+v_k} (x-c) = 
  s_k(x)s_k(x+v_k) = s_k(x)(s_k(x) + s_k(v_k))=(s_k(x))^2+s_k(x) \).
\end{corollary}
Note that $s_i(x)$ w.r.t. any basis of field is linear \cite[Theorem 1]{lch-afft}.
Thus, the vanishing polynomials of tower and Cantor bases are the same.
%
Since the $s_i(x)$ determines the $V_i$ by unique factorization theorem,
we have proved Claim~\ref{lemma:space:equal}.





\subsubsection{Subfield Multiplication in Tower Fields}

We show multiplying by a subfield element is not only cheaper than
general multiplication but also calculated as a vector-scalar product in this section.
The cost depends on the size of the subfield. 
For example, to multiply $a \in \TGF{2^{2w}}$ by $b \in \TGF{2^w}$,
the $a$ is represented as a polynomial $a := a_0 + a_1 x \in \TGF{2^w}[x]$
with $a_0,a_1 \in \TGF{2^w}$.
Hence the product of $a \cdot b \in \TGF{2^{2w}}$ is calculated as two multiplications
in $\TGF{2^{w}}$, i.e., $ (a_0  + a_1 x) \cdot b$.
It is easy to generalize to the following proposition.
\begin{proposition}
\label{lemma:towergfmul}
Given $a \in \TGF{2^{l_1}} = V_{{l_1}} , b \in \TGF{{2^{l_2}}} = V_{{l_2}}$, and $l_2| l_1$, 
$a \cdot b \in V_{{l_1}}$ 
%
can be performed with $l_1/l_2$ 
field multiplications in $\TGF{2^{{l_2}}}$.
\end{proposition}

%


\subsubsection{LCH FFT over Tower Fields}

Besides the compatibility with Cantor allowing efficient \lchbtfy and basis conversion
in tower fields, 
we show the multipliers in \lchbtfy fall to smaller numbers
which can be optimized by subfield multiplication with Prop.~\ref{lemma:towergfmul}
in this section.


Recall that the two multipliers are $s_k(\alpha)$ and $s_k(v_k)$,
which corresponds to $s_k(\beta_k)$ in Alg.~\ref{alg:fft}.

First we have $s_k(v_k) = 1 $ by Prop.~\ref{lemma:subPolyEval}.
We can thus avoid the multiplication of $s_k(v_k)$ as in Cantor basis.
One butterfly unit thus contains one multiplication and two additions.

\begin{figure}[!tbh]
  \centering
    \includegraphics[width=0.5\textwidth]{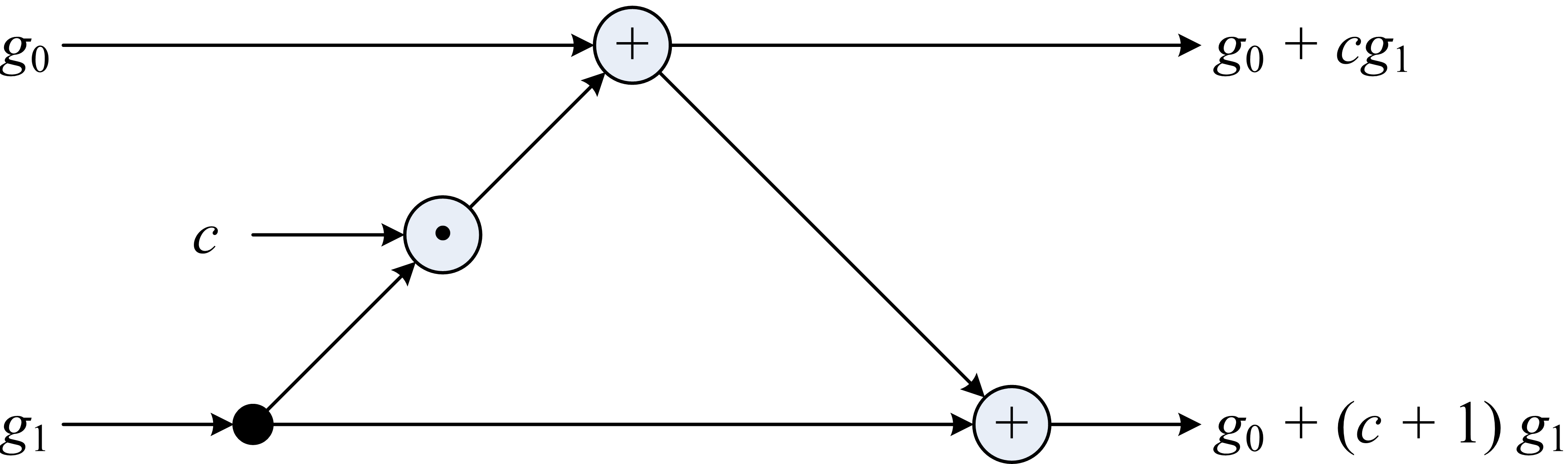}
    \caption{Details of the butterfly unit.}
    \label{fig:butterfly-detail}
\end{figure}
Figure~\ref{fig:butterfly-detail} shows the details of the butterfly unit.
It is also an example of
evaluating a degree-1 polynomial $f(x) = f_0 + f_1 x$ at points $\{c , c+1\}$, 
which $\alpha = c$ with the notation in Alg.~\ref{alg:fft}.
Since $X_1(x) = x$, the degree-1 polynomial after basis conversion is identical
to the original polynomial, i.e., $g(X) = g_0 + g_1 X_1 = f_0 + f_1 X_1$.
The only effective multiplier is $s_0(\alpha) = s_0(c) = c$.

Now we discuss the multiplier $s_k(\alpha)$.
\begin{proposition}
\label{lemma:SPderivative}
For any $v_i = \overline{2^i} $ w.r.t. tower representation, $ v_i^2 + v_i \in V_i$. 
\end{proposition}
\textbf{Proof:} See Appendix~\ref{sec:prooft}.\\

Contrast this with the Cantor basis, where $\beta_i^2 + \beta_i = \beta_{i-1}$.
Note that $s_1: x \mapsto x^2 + x$ is a linear map
from vector spaces $V_{i+1}$ to $V_i$ and its kernel is $\{0,1\}$.
Since each vector in $V_i$ is exactly the image 
of two vectors in $V_{i+1}$,
each vector in $V_i\backslash V_{i-1}$ is also the image
of exactly two vectors in $V_{i+1}\backslash V_i$.
Hence $s_1(v_i) = v_{i-1} + u$ where $u \in V_{i-1}$, i.e.,
$s_1$ shortens the tower representation of $x$ by exactly one bit.
Since 
$s_k$ is just applying $s_1$ consecutively $k$ times,
the multiplier $s_k(\alpha)$ in LCH butterflies is (a)
$k$ bits shorter than $\alpha$ in the tower representation,
and (b) independent of the least significant $k$ bits of $\alpha$.

%
Fig.~\ref{fig:example:butterfly:polymul} depicts the evaluation 
of a degree-$7$ polynomial
at $16$ points $\{ 0,1,\ldots,\mathtt{0xf} \}$ 
with LCH FFT, using $4$ ($= \log 16$) layers of butterflies.
We can observe that  
the multipliers in butterflies are smaller than the actual evaluation points.

\begin{figure}[!htb]
\centering
    \includegraphics[width=0.6 \textwidth]{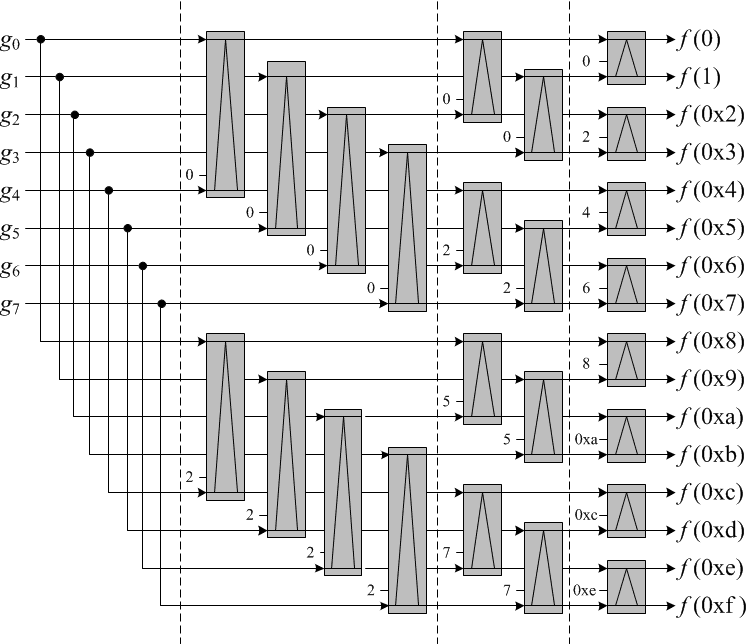}
    \caption{The forward butterfly units for evaluating
    a degree-$7$ polynomial $ f(x) = g(X) = g_0 + \cdots + g_7 X_7$
    at $16$ points $\{ 0,1,\ldots,\mathtt{0xf} \}$. }
    \label{fig:example:butterfly:polymul}
\end{figure}

Since coefficients above degree-$7$ are all $0$,
the first layer 
is simply  a ``fan-out'' 
(in software,  copying a block of memory).
The multipliers in second layer are calculated by evaluating the degree-$4$ $s_2(x)$ at two 
$\alpha \in \{ 0, \mathtt{0x8} \in \GF{16} \} $, resulting in the
 small multipliers $\{ 0 , \mathtt{0x2} \in \GF{4} \} $.
The third layer evaluates the degree-2 $s_1(x)$ at 4 points 
$\alpha \in \{0,\mathtt{0x4},\mathtt{0x8},\mathtt{0xc}\}$
and results in the multipliers $\{0,\mathtt{0x2},\mathtt{0x5},\mathtt{0x7}\}$.
In the last layer, the multipliers are the subfield elements
$ \{ 0, \mathtt{0x2} , \mathtt{0x4} , \mathtt{0x6} , \ldots\mathtt{0xe} \}$
themselves because $s_0(\alpha) = \alpha$.

\subsection{Faster Multiplication using FFTs in \TGF{2^{128}} or \TGF{2^{256}}}
\label{sec:tower:mul}

The main idea of our algorithm is to trade an
expensive field isomorphism for the faster 
multiplication by small subfield elements.

We show the algorithm for multiplying binary polynomials in Alg.~\ref{alg:bitpolymul}.
Assume we use a $w=64$-bit word first.
The input polynomials are in $64$-bit blocks after partitioning(\texttt{Split()}).
We perform the basis conversion before changing representations
so as to have more densely packed data during the conversion. 
Then we convert the data to the tower representation.
The data are kept in the tower representation for performing \lchbtfy.
We efficiently compute the subfield multiplication
with the technique
in Sec.~\ref{sec:gf-tower-mult}.  The \pointmul are also
performed in $\TGF{2^{128}}$; the somewhat different operations are
detailed in Sec.~\ref{sec:gf-general}.  The inverse
butterfly and basis conversion stages are then computed (still in the
tower representation) before converting back to the polynomial
basis of \GF{2^{128}}.  Then we split the $128$-bit results into $64$-bit blocks
and collate coefficients for the final result (\texttt{InterleavedCombine()}).

An initial block width of $w=128$ bits is also possible and in this
case the operative field is \TGF{2^{256}}.
\begin{algorithm}[!tbh]
\SetKwFunction{bpolymul}{\ensuremath{\texttt{binPolyMul}}}
\SetKwFunction{split}{\ensuremath{\texttt{Split}}}
\SetKwFunction{combine}{\ensuremath{\texttt{InterleavedCombine}}}
\SetKwFunction{changerep}{\ensuremath{\texttt{changeRepr}}}
\SetKwFunction{bc}{\ensuremath{\texttt{BasisCvt}}}
\SetKwFunction{nfft}{\ensuremath{\texttt{FFT}_{\texttt{LCH}}}}
\SetKwFunction{ibc}{\ensuremath{\texttt{iBasisCvt}}}
\SetKwFunction{infft}{\ensuremath{\texttt{iFFT}_{\texttt{LCH}}}}
\SetKwInOut{Input}{input}
\SetKwInOut{Output}{output}
\bpolymul{ $a(x),b(x)$ } : \\
\Input{ $a(x) , b(x) \in \GF{2}[x] $ .}
\Output{ $ c(x) = a(x)\cdot b(x) \in \GF{2}[x]$ .}
\BlankLine
$f_a(x) \in \GF{2^{w}}[x] \gets \split( a(x) )$ . \\
$f_b(x) \in \GF{2^{w}}[x] \gets \split( b(x) )$ . \\
$g_a(X) \in \GF{2^{w}}[X] \gets \bc( f_a(x) ) $ . \\
$g_b(X) \in \GF{2^{w}}[X] \gets \bc( f_b(x) ) $ . \\
$ \widetilde{g_a}(X) \in \TGF{2^{2w}}[X] \gets \changerep( g_a(X) ) $ . \\
$ \widetilde{g_b}(X) \in \TGF{2^{2w}}[X] \gets \changerep( g_b(X) ) $ . \\
$[ \widetilde{f_a}(0),\ldots,\widetilde{f_a}( \overline{2n-1} )] \in \TGF{2^{2w}}^{2n} \gets \nfft( \widetilde{g_a}(X) , 0 ) $. \\
$[ \widetilde{f_b}(0),\ldots,\widetilde{f_b}( \overline{2n-1} )] \in \TGF{2^{2w}}^{2n} \gets \nfft( \widetilde{g_b}(X) , 0 ) $. \\
$ [ \widetilde{f_c}(0),\ldots,\widetilde{f_c}(\overline{2n-1}) ] \in \TGF{2^{2w}}^{2n} \gets 
  [\widetilde{f_a}(0) \cdot \widetilde{f_b}(0),\ldots,\widetilde{f_a}(\overline{2n-1})\cdot \widetilde{f_b}(\overline{2n-1}) ] $ \\
$ \widetilde{g_c}(X) \in \TGF{2^{2w}}[X] \gets \infft( [ \widetilde{f_c}(0), \ldots , \widetilde{f_c}( \overline{2n-1} )] ) $ . \\
$ \widetilde{f_c}(x) \in \TGF{2^{2w}}[x] \gets \ibc( \widetilde{g_c}(X) ) $ . \\
$ f_c(x) \in \GF{2^{2w}}[x] \gets \changerep( \widetilde{f_c}(x) ) $ . \\
$ c(x) \in \GF{2}[x] \gets \combine( f_c(x) ) $ . \\
return $ c(x) $.
\caption{Multiplications for binary polynomials.}
\label{alg:bitpolymul}
\end{algorithm}



\section{Implementation}
\label{sec:imple}


\subsection{$\TGF{2^{128}}$ Multiplication by subfield elements}
\label{sec:gf-tower-mult}
The operation for multiplying an element of tower fields by an subfield element is
implemented as a scalar multiplication of a vector over various subfields by a scalar
with Prop.~\ref{lemma:towergfmul}.  We show how to calculate the product
efficiently in current mainstream computers in this section.

\subsubsection{Scalar Multiplication with Vector Instruction Set}
The typical single-instruction-multiple-data (SIMD) instruction set
nowadays is Advanced Vector Extensions 2 (AVX2), providing 256-bit \texttt{ymm} registers
on x86 platforms.
We use the table-lookup instruction \vpshufb in AVX2 for multiplying
elements in $\TGF{2^{128}}$ by subfield elements
which are $\GF{16}$ to $\TGF{2^{32}}$.
We demonstrate the scalar multiplication over subfields with \pshufb instruction
which is the precursor to \vpshufb and uses 128-bit \texttt{xmm} registers.

\pshufb takes two 16-byte sources which one is a lookup table of 16
bytes $\bm{x} = (x_0,\, x_1,\allowbreak \ldots,\allowbreak \, x_{15})$
and the other is 16 indices
$\bm{y} = (y_0,\, y_1,\allowbreak \ldots,\allowbreak \, y_{15})$.  The 16-byte
result of ``\pshufb $\bm{x},\bm{y}$'' at position $i$ is $x_{y_i\!\!\!\mod 16}$ if
$y_i\ge 0$ and $0$ if $y_i<0$.  
\vpshufb simply performs two copies of
\pshufb in one instruction. The two instructions are suitable
for scalar multiplication over small fields \cite{CYC2013Raidq}.

For scalar multiplication over \GF{16},
we first prepare 16 tables; each table stores the product of all elements
and a specific element in $\GF{16}$.
Suppose we have $ \bm{a} \in \GF{16}^{32}$ and $ b \in \GF{16}$, 
we can apply \vpshufb to the prepared ``multiply-by-$b$'' table 
and $\bm{a}$ for the product of $\bm{a}\cdot b$.
Since we use one 256-bit register to store 64 elements in \GF{16},
the data in $\bm{a}$ have to be
split into nibbles (4-bit chunks) before applying \vpshufb.

The scalar multiplication over \GF{256} is similar to \GF{16}
except that the number of prepared tables becomes 256 and one
multiplication costs 2 \vpshufb.

For scalar multiplication over \TGF{2^{16}} or \TGF{2^{32}}, we
implement the field multiplication as polynomial multiplication in $\GF{256}[x]$
because $2^{16}$ or $2^{32}$ prepared tables is too much for caches.
The Karatsuba's method is applied to 
reduce the total number of multiplications in \GF{256}
while multiplying polynomials.

\subsubsection{Transpose the Data Layout for Higher Parallelism}
While multiplying $\TGF{2^{128}}$ elements by elements in 
$\TGF{2^{16}}$ or $\TGF{2^{32}}$, the bytes in a $\TGF{2^{128}}$
element might multiply by different multipliers in $\GF{256}$.
Since we use scalar multiplication over \GF{256} as our building blocks,
multiplying by different multipliers reduces the efficiency.

\paragraph{Example: Product of $\TGF{2^{128}}$ and $\TGF{2^{16}}$ elements}
The natural data layout of $\TGF{2^{128}}$ consists of eight
consecutive $\TGF{2^{16}}$ elements, each with its two bytes stored side
by side. 
Suppose we are multiplying
$\bm{a} = (a_0,\ldots,a_{15}) \in \TGF{2^{128}} := \GF{256}^{16}$ by
$c = (c_0,c_1) \in \TGF{2^{16}} := \GF{256}^2$.
The $\bm{a}$ is naturally stored as
$( (a_0,a_1),(a_2,a_3), \ldots ,\allowbreak (a_{14},a_{15}) ) \in
\TGF{2^{16}}^{8}$.  
To perform the field multiplication in 
$\TGF{2^{16}} := \GF{256}[x]$ with Karatsuba, we then must compute
$ ((a_0,0),\allowbreak (a_2,0),\ldots,\allowbreak(a_{14},0)) \cdot c_0
$ and $ ((0,a_1),\ldots,\allowbreak(0,a_{15})) \cdot c_1 $. In this case,
one has to mask off half the components in $\bm{a}$ and thus
reduces the efficiency.

\paragraph{``Byte-Slicing'' layout:}
To perform many scalar multiplications withs \vpshufb efficiently,
a possible solution is to store each of the
16 bytes in an \TGF{2^{128}} element in a separate register.
This  rearrangement of data layout 
is exactly equivalent to a $16 \times 16$ transposition of a byte matrix.
After each byte in the same position of $\TGF{2^{128}}$ elements is
collected in the same register (cf.~Fig.~\ref{fig:layout2}), 
SIMD instructions can cover entire registers.

\paragraph{Transposition as needed:}
While multiplying two elements by a subfield element instead of
one multiplicand in previous example, e.g.,
$ [\bm{a},\bm{b} \in \TGF{2^{128}}]$ multiply by $c \in \TGF{2^{16}}$,
we only have to split the even and odd bytes in $\bm{a}$ and $\bm{b}$
for performing scalar multiplication efficiently.
Hence, we can just use a $2\times2$ transposition, 
which converts $\bm{a}=(a_0,\ldots,a_{15}) \mbox{ and } \bm{b}=(b_0,\ldots,b_{15})$
to $(a_0,b_0,a_2,b_2,\ldots,b_{14}) \mbox{ and } (a_1,b_1,a_3,\ldots,b_{15})$,
to split the even and odd bytes.
The $2 \times 2$ transposition involves only 2 registers 
instead of 16 registers in a $16\times 16$ transposition
and thus is more efficient. 
For multiplying by an $\TGF{2^{32}}$ element,
we need to $4 \times 4$ transpose our data. 
The following Fig.~\ref{fig:layout2} depicts the process of
data rearrangement on different demand.
\begin{figure}[h]
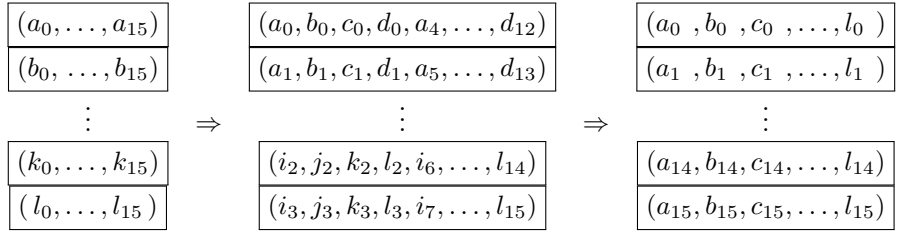

\centering
\[
\begin{array}{ccc c c}
\boxed{ (a_0,\ldots,a_{15}) } & & 
\boxed{ (a_0,b_0,c_0,d_0,a_4,\ldots,d_{12}) } & & 
\boxed{ ( a_{0\,\,\,} , b_{0\,\,\,} , c_{0\,\,\,}, \ldots , l_{0\,\,\,} ) } \\ 
\boxed{ (b_0,\,\ldots,b_{15}) } & & 
\boxed{ (a_1,b_1,c_1,d_1,a_5, \dots , d_{13}) }  & &
\boxed{ ( a_{1\,\,\,} , b_{1\,\,\,} , c_{1\,\,\,}, \ldots , l_{1\,\,\,} ) } \\
\vdots &  \Rightarrow & \vdots & \Rightarrow & \vdots \\
\boxed{ (k_0,\ldots,k_{15}) } & & 
\boxed{ (i_2,j_2,k_2,l_2,i_6,\ldots,l_{14}) } & &
\boxed{ ( a_{14} , b_{14} , c_{14}, \ldots , l_{14} ) } \\
\boxed{ (\,l_0,\ldots,l_{15}\,) } & & 
\boxed{ (i_3,j_3,k_3,l_3,i_7,\ldots,l_{15}) } & &
\boxed{ ( a_{15} , b_{15} , c_{15}, \ldots , l_{15} ) } \\
\end{array}
\]
\caption{``Byte-matrix transpose'' for $\TGF{2^{128}}$ elements:
  The data layout in the middle is 
  for multiplying by  $\TGF{2^{32}}$ elements, 
   and the rightmost layout is for \pointmul.}
\label{fig:layout2}
\end{figure}

To transpose a matrix, we use similar techniques in \cite{hackersdelight},
i.e., transposing  $16\times16$ can be done after transposing $4\times4$. 
While performing the transposition, we first collect elements in one register
with a byte shuffle instruction (\vpshufb).
The interchange of data between different registers is accomplished by
swizzle instructions in AVX2 instruction set. 
We refer readers to 
\cite{intel-isa}\cite{fog:instructions:2017} for more information
on AVX instructions.


\subsubsection{Field Multiplication in \pointmul}
\label{sec:gf-general}
In \pointmul, we need a generic field multiplication
instead of subfield multiplication in \lchbtfy
for data in tower representation.
%
Assuming the data in $\TGF{2^{128}}$ and in the byte-slice layout,
we need the pointwise field multiplication in \pointmul, 
in contrast to scalar multiplication in \lchbtfy.
Since the byte-slice layout,
the multiplication in $\TGF{2^{128}}$ can be easily performed as
polynomial multiplication in $\TGF{2^{64}}[x]$ in the SIMD manner.
In other words, we reduce one parallelized multiplication 
in $\TGF{2^{128}}$ to several parallelized multiplications in $\TGF{2^{128}}$. 
The process is recursively applied until the pointwise
multiplication in $\GF{16}$.
We then apply the method in \cite{cryptoeprint:2017:636}, which
use the logarithm/exponential tables for SIMD field arithmetic,
for the pointwise multiplication in $\GF{16}$.

\subsection{Revisiting the Choice of Base Field}

We discuss the choice of working field based on
the cost of field multiplication.
While multiplying $d$ bits polynomials, we split the data
into $w$-bit wide chunks and then apply FFT with $l=2w$-bit base field.
Here we want to decide on a suitable size of base field.
In the typical x86 CPU, one may choose $l$ to be $64$,$128$, or $256$.

\paragraph{Irreducible polynomial constructed field:}
In this construction, the field multiplication is implemented with \pclmul instruction.
Using Karatsuba's method and linear folding for the reduction,
one multiplication costs $(1+2)=3$, $(3+2)=5$, and $(9+4)=13$ \pclmul 
in $\GF{2^{64}}$, $\GF{2^{128}}$ and $\GF{2^{256}}$ respectively. 
The number of butterflies in \lchbtfy are $\frac{d}{w}\log \frac{2d}{w}$.
With one field multiplication for each butterfly,
we conclude that working on $\GF{2^{128}}$ has lowest number of \pclmul for \lchbtfy.

\paragraph{Tower field:}

For $l$-bit tower field, the costs of multiplying by subfield
elements in \lchbtfy or generic elements in \pointmul is proportional to length of 
field $l$ as describing in Prop.~\ref{lemma:towergfmul}.
Since number of
multiplications in \lchbtfy proportions to $\frac{d}{l} \log \frac{d}{l}$, 
to enlarge $l$ by two results in
one layer less ($\propto \log \frac{d}{l}$) in \lchbtfy
but the same cost for each layer of butterflies($\propto l \cdot \frac{d}{l}$).

The cost of generic multiplication in \pointmul is $\propto l^{1.7}$
using recursive Karatsuba in Sec.~\ref{sec:gf-general}.
If we compare $l=l_1$ with $l=2l_1$, the smaller field is more efficient.
However, the length of field should fit the underlying machine
architecture for better efficiency. 
In the Intel Haswell architecture with AVX2 instruction
set, it is the most efficient to use 256 bits as a unit
because memory access is the best with 256-bit alignment.
Therefore, we choose $l=256$ bits for Intel Haswell in our implementation.


\subsection{Field Isomorphism}
\label{sec:cal-isomorphism}
The change of field representations is simply a matrix product for a pre-defined
matrix $\mathbf{I}$ with the data $\alpha \in V_k$ as a vector.
We compute the product $\mathbf{I} \cdot \alpha $ with
the famous method of four Russians(M4R)\cite{Aho:1974:DAC:578775}.

With M4R of $l$-bit, one first prepars all possible products of
$\mathbf{I}$ and $l$-bit vectors, i.e., prepares $\mathbf{I}\cdot b$ for all $b \in V_l$
, for all $b \in \mathrm{span}(v_l,\ldots,v_{2l-1})$, \ldots ,
for all $b \in \mathrm{span}(v_{k-l},\ldots,v_{k-1})$.
To compute $\mathbf{I}\cdot \alpha$, one splits $\alpha$ to $l$-bit chunks,
looks up the prepared tables for various segments of $\alpha$, and combines the results.
The number of operations proportions to bit-length of $\alpha$.

The choice of $l$ depends on the size of cache for efficiency.
The size of $\mathbf{I}$ for $128$-bit field is $128 \times 128$ bits, and
L1 cache is $32$ KiB for data in the Intel Haswell architecture.
Therefore, we choose the M4R with $4$-bit, which results in
$ 32 \text{KiB} = 16 \times 128 \times 128$ bits prepared tables.

\subsection{Calculation of Multipliers in the Butterflies}
\label{sec:calc-ska}

Suppose we want to calculate $s_i(\alpha)$ for $\alpha \in V_k$ in a tower field. 
$s_i(\alpha)$ can be calculated recursively via $s_i(\alpha) = s_{i-1}(s_1(\alpha))$. 
Since $s_1$ is a linear map, 
we can prepare a table $ \mathbf{S1} := 
[ s_1(v_1) , \ldots , s_1(v_k) ] $ for the images of $s_1$ on all basis elements $(v_1,\ldots,v_k)$.
Since $\alpha = \sum b_i \cdot v_i $ with $b_i \in \{ 0,1\}$,
we have $s_1(\alpha) = \mathbf{S1} \cdot \alpha$ which is a matrix production.
For further optimization, we can omit the least $i$ bits
of $\alpha$ While evaluating $s_i(x)$ by Prop.~\ref{lemma:SPderivative}.

In our implementations, we actually precomputed 31 tables 
$\mathbf{S1},\mathbf{S2},\ldots,\mathbf{S31}$ for the 
evaluations of $s_1(x),s_2(x),\ldots,s_{31}(x)$ at up to $2^{32}$
points and avoid the need for recursion.
Again, we use M4R of $8$-bit to accelerate matrix-vector products. 
Querying a result of $32$-bit $\alpha$ costs $4$ table lookups and
the storage for tables is 4 KiB($4 \times 2^8 \times 32$ bits),
which fits our L1 cache.

\section{Results and Discussion}
\label{sec:results}
We benchmark our implementation\footnote{
Software can be downloaded from
\url{https://github.com/fast-crypto-lab/bitpolymul} .}
on the Intel Haswell architecture (same as \cite{DBLP:conf/issac/HarveyHL16}).
Our hardware is Intel Xeon E3-1245 v3 @3.40GHz with turbo boost disabled
and 32 GB DDR3@1600MHz memory.
The experiments was run in ubuntu 1604, Linux version 4.4.0-78-generic and
the compiler is gcc: 5.4.0 20160609 (Ubuntu 5.4.0-6ubuntu1~16.04.4).

We show our main results in Tab.~\ref{tab:bitpolymul},
comparing against \cite{gf2x} and \cite{DBLP:conf/issac/HarveyHL16}.
Three of our implementations are denoted by their base fields in \lchbtfy.
The version of $\GF{2^{128}}$ is the simple version with evaluation points in Cantor basis
in Sec.~\ref{sec:simple:mul}.
The versions of $\TGF{2^{128}}$ and $\TGF{2^{256}}$ are tower field implementations in Sec.~\ref{sec:tower:mul}.
In general, our implementations are around $10\%$ to $40\%$ faster than the other 
binary polynomial multipliers.

\begin{table}[h] 
\begin{center}
\begin{threeparttable}
\caption{Products in degree $< d$ in $\GF{2}[x]$ on Intel Xeon E3-1245 v3 @ 3.40GHz ($10^{-3}$ sec.)
}
\label{tab:bitpolymul}
\begin{tabular}{|l | r | r | r| r | r | r | r | r | r | }
\hline
   $\log_2 d/64$ & 15 & 16  & 17 &  18 & 19 & 20 & 21 & 22 & 23   \\  \hline\hline
This work, $\TGF{2^{256}}$               &  9 & 19 & 40 &  90 & 212 & 463 & 982 & 2050 & 4299  \\  \hline 
This work, $\TGF{2^{128}}$          & 11 & 22 & 48 & 104 & 243 & 527 & 1105 & 2302 & 4812  \\  \hline 
This work, $\GF{2^{128}}$      & 12 & 26 & 55 & 119 & 261 & 554 & 1181 & 2491 & 5282 \\  \hline \hline
\texttt{gf2x} \cite{gf2x} \tnote{a}  & 12 & 26 & 59 & 123 & 285 & 586 & 1371 & 3653 & 7364  \\ \hline 
$\GF{2^{60}}$ \cite{DBLP:conf/issac/HarveyHL16} \tnote{b} & 14 &  29 &  64  &  148   & 279   & 668   &  1160  & 3142 & 7040 \\  \hline 
\hline
\end{tabular}
\begin{tablenotes}
  \small
  \item [a] Version 1.2.  Available from \url{http://gf2x.gforge.inria.fr/}
  \item [b] SVN r10663. Available from \url{svn://scm.gforge.inria.fr/svn/mmx}
\end{tablenotes}
\end{threeparttable}
\end{center}
\end{table}

The first notable result is that our simple version over
$\GF{2^{128}}$ outperforms the previous implementations in \cite{gf2x} and
\cite{DBLP:conf/issac/HarveyHL16} for polynomials over $2^{17+6}$
bits.  The result shows that additive FFT are better based FFT for multiplying
binary polynomials given the same  multiplication in the
base field using \pclmul.

%
%
%
We can also compare the effectiveness of field multiplication
from second and third rows.
The comparison between $\TGF{2^{128}}$ and
$\GF{2^{128}}$ shows that the subfield multiplication with \vpshufb
outperforms generic multiplication with \pclmul.
This is a counter-intuitive result  since
\pclmul is dedicated to multiply binary polynomials by design.
The effect can be shown qualitatively: we can do $32$
multiplications in \GF{256} using 2 \vpshufb's.  Multiplying by
\TGF{2^{32}} elements (the largest multipliers in the butterflies) costs 9
\GF{256} multiplications.  Multiplying an $\TGF{2^{128}}$ element by
an \TGF{2^{32}} element is 4 \TGF{2^{32}} multiplications, so each
\TGF{2^{32}} to \TGF{2^{128}} product takes $9/4$ \vpshufb's on
average.  Similarly, a \TGF{2^{16}} to \TGF{2^{128}} product takes $3/4$
\vpshufb's.  The average is less than 2 \vpshufb's compared to $5$
\pclmul's when using \GF{2^{128}}.


\paragraph*{The profiles of various components}
The ratio of relative cost in our fast subfield multiplication 
between basis conversion $:$ butterfly process $:$\ \pointmul$:$
change of representations for multiplying $2^{20+6}$-bit polynomials are
 1: 3.06: 0.08: 0, 1: 2.11: 0.27: 0.54, and 0.73: 1.74: 0.47: 0.57
for $\GF{2^{128}}$, $\TGF{2^{128}}$ and $\TGF{2^{256}}$, respectively.

The results show that the change of tower representation increases the
cost in \pointmul besides the cost itself. However, the efficient subfield
multiplication reduces the cost of butterfly process and the effects
are greater than the cost increased. The version of $\TGF{2^{256}}$ can
even reduce the cost of basis conversion because the better aligned
memory access fits into machine architecture.

More results on newer Intel architecture and profiles can be found in Appendix~\ref{sec:profile:rawdata}.

\paragraph*{Discussion about the most recent result}
In \cite{vanderhoeven:hal-01579863}, Hoeven et al. present a new result
of multiplying binary polynomials.
They use traditional multiplication FFT and
 replace the Kronecker segmentation in Sec.~\ref{sec:mul-seg-fft} with a Frebenius DFT
for saving about $50\%$ of running time of $\GF{2^{60}}$\cite{DBLP:conf/issac/HarveyHL16}.
We note the same technique can be applied to additive FFT as well as about $50\%$ saving of cost.
The method for additive FFT is described in Appendix~\ref{sec:frobenius:bijection}.

\subsection{Truncated FFT with Non-Power-of-Two Terms}

The experiments are actually the optimal cases for multiplying polynomials
with power-of-two terms by additive-FFT-based multipliers.

Since the cost of additive-FFT over number of terms of polynomial is
highly stairwise, one can truncate the FFT to make the curve somewhat
smoother as in \cite{gf2x}.  One truncated version of additive FFT for $n=3\cdot 2^k$ was
shown in \cite{cryptoeprint:2017:636}.


\subsection{Further Discuss on Other Possible Implementations}

In this section, we discuss two possible variants of implementations.
\paragraph*{Tower field implementations with evaluation points in Cantor basis}
If we choose evaluation points in Cantor basis, the evaluation 
of $s_i(\alpha)$ seems faster than the M4R technique in previous section.
However, we still need to change the representation of $s_i(\alpha)$ from
Cantor basis to tower field in this case. This operation results in the same cost
as the calculation of $s_i(\alpha)$ in tower field with M4R.

\paragraph*{Performing subfield multiplication in Cantor basis}
For subfield multiplications in tower field, the efficiency comes from the powerful
\vpshufb instruction.
Since we don't have Prop.~\ref{lemma:towergfmul} in Cantor basis, we can not use
\vpshufb in Cantor basis in the same way of tower fields.
Another options for implementing field multiplication in Cantor basis is to 
use bit-sliced data and logical instructions \cite{auth256}.
There are at least 64 \texttt{ymm} registers for operating in 64-bit base field,
resulting inefficiency from too much data in play.

\section{Concluding Remarks}
\label{sec:summary}

We have presented our efficient multipliers based on recently developed
additive FFTs, which has similar but lower multiplicative complexity
as the ternary variant of Sch\"{o}nhage's algorithm used in
\cite{gf2x}. 

%

In \cite{gf2x}, Brent \emph{et al.} also implemented the
Cantor~\cite{Cantor:1989} algorithm beside Sch\"onhage.
They
concluded that \emph{``Sch\"onhage's algorithm is consistently faster by
  a factor of about 2 (than Cantor)''}. 
In their implementation, multiplicative FFT outperformed additive FFT.

Our experiments show that recently developed additive-FFT does help
for multiplying binary polynomials of large degrees in practice.  
We derive a further
advantage by exploiting the lower cost of multiplying by subfield
elements in tower fields.

\paragraph{Future Work:}
Our implementation is written in C and 
may be improved with assembly for better register allocations.
Further, AVX-512 instructions, featuring 512-bit SIMD instructions,
will be more widely available soon.  We probably
cannot speed up the additive FFT multiplier to a factor of two, but
surely AVX-512 can be expected to offer a substantial advance.







\newcommand{\etalchar}[1]{$^{#1}$}
\providecommand{\skiptext}[1]{} \hyphenation{ASIA-CRYPT}

\appendix

\section{Claims about $s_i(x)$ in Tower Fields and Proofs}
\label{sec:prooft}


\textbf{Proposition~\ref{lemma:subPolyEval}.}
%
$s_k( v_k ) :=\prod_{b\in V_k}(v_k-b) = 1$.\smallskip

We first discuss the special case that $k = 2^l$ is power of two and
there is a $\GF{q}$ for $q=k$.
By Galois's theory, we have $s_k(x) = \prod_{b\in \GF{q}} (x-b) = x^q + x$.
\begin{proposition}
  If $q $ is a power of two, and choose any $a\in\GF{q}$ such that
  $\GF{q^2} := \GF{q}[x_l]/(x_l^2+x_l+a)$ is a valid field extension, then
  \[  \prod_{b\in \GF{q}} (x_l-b) = x_l^q + x_l = 1.\]
\end{proposition}
\begin{proof}
\begin{eqnarray*}
   x_l^q + x_l &=& \left(x_l^2+x_l\right) + \left(x_l^2+x_l\right)^2 + \left(x_l^2+x_l\right)^4 + \cdots + \left(x_l^2+x_l\right)^{q/2} \\
  &=& a + a^2 + a^4 + \cdots + a^{q/2} \\
  &=& \text{Trace of } a \text{ (in \GF{q} over \GF{2})} \in \{0,1\}. 
\end{eqnarray*}
But zero here would be contradictory because  all $q$ solutions of
$x^q=x$ are already in \GF{q}, but $x_l$ is in $\GF{q^2} \setminus \GF{q}$,
hence we must have $x_l^q+x_l=1$. 
\end{proof}

\begin{proof}[Proof of Prop.~\ref{lemma:subPolyEval}]
If $q = 2^j$ is a power of two, then $V_q = \GF{2^q}$ and the result
holds according to the Proposition above.  So we assume that the
proposition holds for all $k<h$ and $h = q + \ell < 2q$, where
$q=2^j$, and note that for $a \in V_q = \GF{2^q}$, we have
\(\left(v_q a\right)^q + v_q a = (v_q +1)a + v_q a = a\), and
\begin{align*}
\prod_{b\in V_h}(v_h-b) &=\prod_{c\in \mathrm{span}(v_q,\ldots v_{h-1})} \left( \prod_{c'\in V_q} (v_h + c + c')\right) & \mbox{ ( divide $V_h$ ) } \\
&= \prod_{c''\in V_{h-q}} \left( \prod_{c'\in V_q}\left( v_q (v_{h-q} + c'') + c'\right)\right)  & \mbox{(replace $c$ by $v_q c''$)} \\
&= \prod_{c''\in V_{h-q}}  \left[ \left( v_q ( v_{h-q} + c'') \right)^q + v_q ( v_{h-q} + c'' )\right] & \mbox{(Galois's theory)} \\
 &=  \prod_{c''\in V_{h-q}} \left(v_{h-q} + c''\right) = 1. & \mbox{(by induction)}
\end{align*}
\end{proof}

%
%

\noindent
{\bf Proposition~\ref{lemma:SPderivative}.}
For any $v_i = \overline{2^i} $ w.r.t. tower representation,
$v_i^2 + v_i \in V_i$. \smallskip\\

In the proof of Prop.~\ref{lemma:SPderivative},
we need the following consequences from definition~\ref{def:v:basis} and \ref{def:novelpoly}.
\begin{corollary}
  If $q = 2^{k} > i$, then $x_{k+1} v_i = v_q v_i = v_{2^k + i}$.
\end{corollary}
\begin{corollary}
\label{cons:aboutSpace}
  If $q = 2^{k} > i$ and $v\in V_i$,  then $x_{k+1} v = v_q v\in V_{2^k + i}$.
\end{corollary}

\begin{proof}[Proof of Prop.~\ref{lemma:SPderivative}]
If $i = 2^k$, then $v_i = x_{k+1}$, and
$v_i^2 + v_i = x_1 x_2 \cdots x_k = v_{2^k-1}$ by the defintion of the
$x_i$ so the claim holds.  Using mathematical induction, we let
$2^k > i = 2^{k-1} + j$.  Thus $v_i = x_k v_j$, and
$v_i^2 + v_i = (x_k^2 + x_k) v_j^2 + x_k (v_j^2 + v_j)$.  The first
term is the product of two terms in $\GF{2^{2^{k-1}}}$ and therefore
is itself in $\GF{2^{2^{k-1}}}= V_{2^{k-1}} \subseteq V_i$.  The
second is the product of $v_{2^{k-1}}$ and an element of $V_j$ (by the
induction hypothesis), which is by Corollary~\ref{cons:aboutSpace} also in
$V_{2^{k-1}+j} = V_i$.
\end{proof}

\section{Basis Conversion: An Example}
\label{appendix:basis-conversion}

\begin{figure}[htbp]
\centering
    \includegraphics[height=6cm]{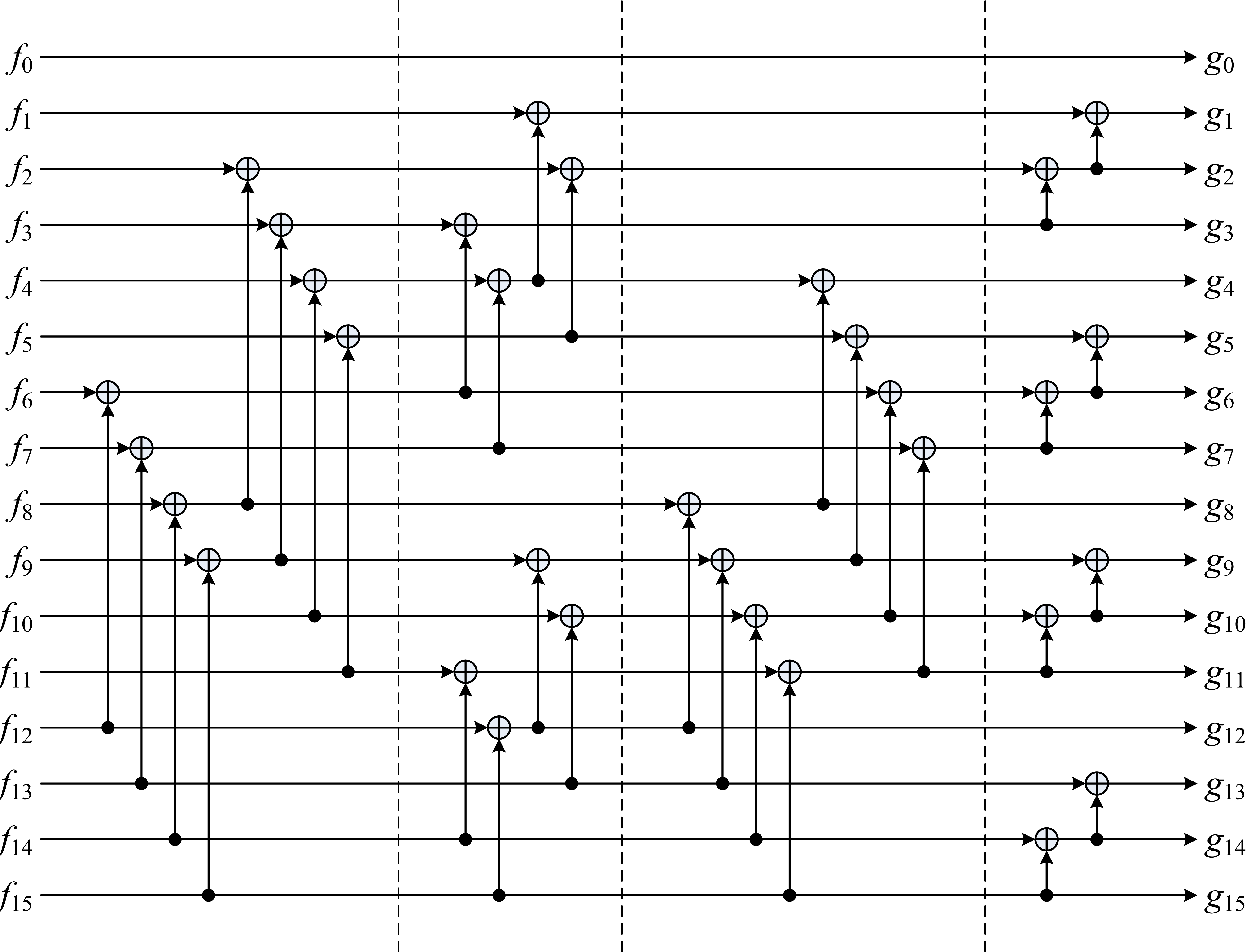}
\caption{From $f(x)=f_0 + \cdots + f_{15}x^{15} $ to $g(X) = g_0 + \cdots + g_{15} X_{15}$ by  Algorithm~\ref{alg:changeBasis2}.}
\label{fig:example:optbscvt}
\end{figure}

Figure~\ref{fig:example:optbscvt} shows an example for converting a
degree-15 polynomial to \novelpoly basis.
We have to divide by 3 different $s_i(x)$'s, namely $s_1(x),s_2(x)$,
and $s_3(x)$, which has 4 terms.
One can see there are actually 4 layers of division in
Figure~\ref{fig:example:optbscvt} and the number of \xor's are the
same in all layers.

What the first 2 layers perform is variable substitution into terms of
$ y = s_2(x) = x^4+x$ (See algorithm~\ref{alg:varSubs}).
The first layer is a division by $s_2(x)^2 = x^8 + x^2$ and second
layer is divisions by $s_2(x) = x^4 + x $ on the high
degree and low degree polynomials from the first layer.
The third layer is a division by $s_3(x)= y^2 + y$ \emph{by adding
  between coefficients of terms differing by a factor of $s_2(x)$, and
  4 positions apart} (see 3rd column of Table~\ref{tab:si}).
The last layer divides by $s_1(x) = x^2 + x$ for 4 short polynomials
(Last loop in Algo.~\ref{alg:changeBasis2}).
Note that we only do division by two terms polynomials in the conversion.

\section{Frobenius bijection between $\GF{2}[x]_{<128\cdot n}$ and $\GF{2^{128}}^n$}
\label{sec:frobenius:bijection}
In this section, we perform an FFT directly on $\GF{2}[x]$ instead of combining many coefficients
to a larger finite field, i.e., abandoning the Kronecker segmentation.
We refer the reader to \cite{vanderhoeven:hal-01579863} for preliminaries and focus on 
performing the idea with additive-FFT in this section.

\subsection{The partition of evaluated points}

Given $\phi_2$ is the Frebenius operation in $\GF{2^{128}}$, which is exactly to square 
an element in an extension field of characteristic 2.
Thus, we have $\phi_2(\beta_0) = \beta_0$, and $\phi_2(\beta_i) = \beta_i^2 = 
\beta_i + \beta_{i-1}$ for $i > 0$ in Cantor basis.
Moreover, the order of $\phi_2$ for $\beta_{i>0}$ is $ 2 \cdot 2^{ \lfloor \log_2 i \rfloor} $.

\begin{definition}\label{def:eval:in:gf2}
Given $a(x) \in \GF{2}[x]_{< 128 \cdot n}$ and $n = 2^m$ with $m <63$ , 
define a linear map $ E_{\beta_{m+64} + V_m} : \GF{2}[x]_{< 128 \cdot n} \rightarrow \GF{2^{128}}^n$ 
is the evaluations of $a(x)$ at points $\beta_{m+64} + V_m$, i.e.,
$ E_{\beta_{m+64} + V_m}: a(x) \mapsto \{ a( \beta_{m+64} + i) : i \in V_m \} $ .
\end{definition}

We can see the evaluation of $a(x)$ at $n = 2^m$ points which is less than the number of coefficients.
However, the evaluation points are in $\GF{2^{128}}$.
It can be checked the order of $\phi_2(\beta_{m+64})$ is $128$.

\begin{proposition}
\label{lemma:the:partition}
$E_{\beta_{m+64} + V_m}$ is a bijection between $\GF{2}[x]_{<128\cdot n}$ and $\GF{2^{128}}^n$.
\end{proposition}

\subsection{Truncated Additive FFT}

To evaluate a polynomial $a(x) \in \GF{2}[x]_{<128\cdot n}$ at $2^{m}$ points in $\GF{2^{128}}$,
we treat each coefficient of $a(x)$ is an element in $\GF{2^{128}}$.
Since there are $128 \cdot n = 128 \cdot 2^m$ coefficients and only $2^m$ evaluation points,
the $E_{\beta_{m+64} + V_m}$ is exactly a truncated additive FFT to $1/128$ points of
an FFT at points $\beta_{m+64} + V_{m+7}$.
The basis conversion can be proceeded straightforwardly with alg.~\ref{alg:changeBasis2}. 
With the notation of alg.~\ref{alg:fft},
the evaluation of $a(x)$ at $\beta_{m+64} + V_m$ can be accomplished
with the scalar displacement $\alpha = \beta_{m+64}$ and reserve only
first $1/128$ results of the evaluation points $V_{m+7}$.

\subsection{Encoding: the First Seven Layers of Truncated FFT}

We can also translate the truncated FFT to an ``encoding'', corresponding to 
the starting 7 layers of butterflies, of input coefficients and followed
by a normal \lchbtfy over $\GF{2^{128}}[x]_{<n}$ at points $V_{m}$ with $\alpha = \beta_{m+64}$.
Since there is only one bit for each element in the first layer of the truncated FFT,
each element of 8th layer is effected by $2^7$ bits in the first layers 
The multipliers from the first to 7th layers are $( s_{m+i}(\beta_{m+64}) )_{i=7}^1$.
Hence, we have the ``encoding'' for the $i$-th element is 
$\sum_{j=0}^{127} a_{j\cdot n + i} \cdot ( \prod_{\text{$k$-th bit of $j$ is $1$}} s_{m+7-k}(\beta_{m+64}) )$.
Every element in 8th layer of \lchbtfy is a result of 
linear map with input $( a_{j\cdot n + i })_{j=0}^{127}$ which is a $128 
\times n$ transpose of input data.

\subsection{Results}

Table~\ref{tab:frobenius:bitpolymul} shows the results of our implementation
on multiplying polynomials with additive FFT and the comparison with implementation
from \cite{vanderhoeven:hal-01579863}.

\begin{table}[h] 
\begin{center}
\begin{threeparttable}
\caption{Multiplicaitons with Frobenius partitions.
Products in degree $< d$ in $\GF{2}[x]$ on Intel Xeon E3-1245 v3 @ 3.40GHz ($10^{-3}$ sec.)
}
\label{tab:frobenius:bitpolymul}
\begin{tabular}{|l | r | r | r| r | r | r | r | r | r | }
\hline
   $\log_2 d/64$ & 15 & 16  & 17 &  18 & 19 & 20 & 21 & 22 & 23   \\  \hline\hline
add FFT, $\GF{2^{128}}$ \tnote{a} & 7 & 15 & 32 & 67 & 144 & 319 & 706 & 1495 & 3139 \\  \hline
DFT, $\GF{2^{60}}$ \cite{vanderhoeven:hal-01579863} \tnote{c} & 6 &  15 &  32  &  72   & 165   & 311   &  724  & 1280 & 3397 \\  \hline
\hline
\end{tabular}
\begin{tablenotes}
  \small
  \item [c] Version 2ae2461. \url{https://github.com/fast-crypto-lab/bitpolymul}
  \item [c] SVN r10681. Available from \url{svn://scm.gforge.inria.fr/svn/mmx}
\end{tablenotes}
\end{threeparttable}
\end{center}
\end{table}

\section{Profiles: Raw Data}
\label{sec:profile:rawdata}

Table~\ref{tab:prof:fft} shows the raw data of various componets of our implementations.
Table~\ref{tab:bitpolymul2-skylake} shows the performance of our implemetations in 
Intel Skylake architecture.

\begin{table}[h]
\centering
\scriptsize
\caption{Running time of various Components, in $10^{-6}$ sec.}
\label{tab:prof:fft}
\begin{tabular}{| ll| rrrrrrr |}
\hline 
\multicolumn{2}{|c|}{$\log_2 d/64$}  & 15   & 16  & 17   & 18      & 19    & 20    & 21  \\ \hline
\hline
\multirow{4}{3em}{$\TGF{2^{256}}$}
 & \texttt{chRepr} $\times2$ & 971  & 1955 & 3482 & 7226  & 13991 & 27550  & 53982 \\ \cline{3-9}
 & \bscvt $\times2$          & 823  & 1846 & 4055 & 8398  & 17598 & 35611  & 80913 \\ \cline{3-9}
 & \lchbtfy $\times2$        & 2054 & 4358 & 9559 & 24004 & 70109 & 153674 & 328211 \\ \cline{3-9}
 & \pointmul                 & 1870 & 3732 & 7591 & 15686 & 33122 & 62471  & 125001 \\ \cline{3-9}
 & \ilchbtfy                 & 923  & 1960 & 4402 & 11782 & 34244 & 75693  & 161648 \\ \cline{3-9}
 & \ibscvt                   & 1081 & 2418 & 4958 & 11008 & 26643 & 60948  & 131862 \\ \cline{3-9} 
 & \texttt{ichRepr}          & 1500 & 2989 & 5969 & 11913 & 24046 & 47632  & 94949 \\ \hline

\multirow{4}{3em}{$\TGF{2^{128}}$}
 & \texttt{chRepr} $\times2$  & 924  & 1866 & 3890  & 7172  & 13237 & 25965  & 52275 \\ \cline{3-9}
 & \bscvt $\times2$           & 1340 & 2916 & 5969  & 12328 & 25415 & 51969  & 113587 \\ \cline{3-9}
 & \lchbtfy $\times2$         & 2613 & 6039 & 13096 & 31852 & 83316 & 187370 & 402506 \\ \cline{3-9}
 & \pointmul                  & 1079 & 2147 & 4305  & 9003  & 19549 & 35978  & 72710 \\ \cline{3-9}
 & \ilchbtfy                  & 1216 & 2778 & 6305  & 15370 & 40224 & 91015  & 195958 \\ \cline{3-9}
 & \ibscvt                    & 1668 & 3377 & 7075  & 15282 & 36113 & 78775  & 174826 \\ \cline{3-9} 
 & \texttt{ichRepr}           & 1417 & 2794 & 5587  & 11202 & 22472 & 44970  & 90738  \\ \hline

\multirow{3}{3em}{$\GF{2^{128}}$}  
 & \bscvt $\times2$    & 1337 & 2910  & 5970  & 12355 & 25118  & 52793  & 112762 \\ \cline{3-9}
 & \lchbtfy $\times2$  & 5856 & 12531 & 26327 & 57631 & 126342 & 272165 & 615555 \\ \cline{3-9}
 & \pointmul           &  317 & 635   & 1297  &  2863 & 7250   & 10278  & 21188 \\ \cline{3-9}
 & \ilchbtfy           & 2766 & 5895  & 12600 & 27338 & 60891  & 131325 & 277294 \\ \cline{3-9}
 & \ibscvt             & 1656 & 3371  & 7044  & 15273 & 36189  & 79076  & 172195 \\ \hline
\end{tabular}
\end{table}

\begin{table}[h] 
\begin{center}
\begin{threeparttable}
\caption{Products in degree $< d$ in $\GF{2}[x]$ on Intel Xeon E3-1275 v5 @ 3.60GHz ($10^{-3}$ sec.)
}
\label{tab:bitpolymul2-skylake}
\begin{tabular}{|l | r | r | r| r | r | r | r | r | r | }
\hline
   $\log_2 d/64$ & 15 & 16  & 17 &  18 & 19 & 20 & 21 & 22 & 23   \\  \hline\hline
This work, $\TGF{2^{256}}$          &  8 & 16 & 34 &  74 & 175 & 382 & 817 & 1734 & 3666  \\  \hline 
This work, $\GF{2^{128}}$  &  9 & 18 & 38 &  84 & 187 & 408 & 889 & 1907 & 4075 \\  \hline \hline
\texttt{gf2x} \tnote{a} \cite{gf2x} & 11 & 23 & 51 & 111 & 250 & 507 & 1182 & 2614 & 6195  \\ \hline 
$\GF{2^{60}}$\cite{DBLP:conf/issac/HarveyHL16}\tnote{b}  & 10 &  22 &  51  &  116   & 217   & 533   &  885  & 2286 & 5301 \\  \hline
\hline
\end{tabular}
\begin{tablenotes}
  \small
  \item [a] Version 1.2. Available from \url{http://gf2x.gforge.inria.fr/}.
  \item [b] SVN r10663. Available from \url{svn://scm.gforge.inria.fr/svn/mmx}
\end{tablenotes}
\end{threeparttable}
\end{center}
\end{table}



\end{document}